\documentclass[suppldata]{interact} 
\usepackage{lmodern}
\usepackage[font=small,labelfont=bf,labelsep=none]{caption}
\usepackage{subcaption}
\usepackage[onehalfspacing]{setspace} 
\usepackage[natbibapa,nodoi]{apacite} 
\usepackage{amsmath,amssymb,amsthm,commath,mathtools}
\usepackage{etoolbox}
\usepackage{enumitem}

\usepackage{hyperref}
\usepackage{cleveref}
\usepackage{xcolor}
\newtheoremstyle{custom}
  {\topsep} 
  {\topsep} 
  {\itshape} 
  {} 
  {\bfseries} 
  {} 
  {.5em} 
  {} 
\theoremstyle{custom}
\newtheorem{theorem}{Theorem}[section]

\newtheorem{corollary}[theorem]{Corollary}
\AtBeginEnvironment{proof}{\vspace*{-2ex}} 
\begin{document}
\title{Agent-based dynamics of criminal propensity}
%
\author{
\name{Daniel Holland\textsuperscript{a}, Abigail Saynor\textsuperscript{a}, Evelyn Svingen\textsuperscript{b} and Galane Luo\textsuperscript{a$\dagger$}}
\affil{\textsuperscript{a}School of Mathematics, University of Birmingham, Birmingham, B15 2TT, UK.\\
\textsuperscript{b}School of Social Policy and Society, University of Birmingham, Birmingham, B15 2TT, UK.}
\thanks{$\dagger$ Corresponding author: g.j.luo@bham.ac.uk}
}
\maketitle
\begin{abstract}
The Retribution and Reciprocity Model (RRM) is a novel framework presented in evolutionary criminology to understand causes of crime through the lens of cooperation. We formalise RRM as an agent-based dynamical system in which criminal propensity evolves through pairwise interactions and through observations of others' interactions. The mathematical model extends bounded confidence opinion dynamics with two novel features: an agent-specific attractor representing the perception of the environment, and the participation of third-party witnesses. We prove that individual propensities remain bounded and establish sufficient conditions for convergence to an extreme value and to the agent's perception of the environment. Simulations show that the dominant system behaviour is not convergence but persistent oscillations, a rare phenomenon in bounded confidence models arising here without the adaptive confidence bounds through which it has previously been obtained. We derive a necessary condition for oscillations, and find that a population sharing a uniform but neutral perception of the environment tends to polarise to the extremes, while a uniform non-neutral perception tends to yield consensus. The polarising regime is confined to a narrow band of near-neutral perceptions, whose width decreases with population size and increases with the strength of reciprocal and retributive tendencies.
\end{abstract}

\clearpage
\section{Introduction} \label{sec:intro}

Criminology seeks to explain why people break rules, and has pursued that goal largely by identifying factors statistically associated with offending. Its major traditions locate the causes of crime variously in blocked aspirations and aversive conditions \citep{merton1938anomie,agnew1992empirical}, in learning and reinforcement within immediate social groups   \citep{sutherland1947differential,akers1998adolescent}, in the strength of social bonds and individual self-control \citep{hirschi1969hellfire,hirschi1990substantive}, in the situational convergence of motivated offenders and opportunities \citep{cohen1979social}, in developmental trajectories across the life course \citep{moffitt1993adolescence,laub1993turning}, and in the interaction of neurological and evolved dispositions with social environments \citep{glenn2013antisocial,walsh2009biosocial}. A persistent criticism of this body of work is that establishing which factors predict crime is not the same as specifying the process by which crime is produced \citep{wikstrom2007role}.

A distinct strand takes that criticism as its point of departure. Situational action theory reorganises the explanation of crime around the interaction between individual propensity and setting, arguing for a criminology built on causal mechanisms rather than risk factors \citep{wikstrom2006individuals,wikstrom2012breaking}, and the broader methodological case is made by analytical sociology \citep{hedstrom2005dissecting,hedstrom2010causal} and by the generative principle that a collective pattern is explained only once it can be grown from the interactions that produce it \citep{schelling1971dynamic,epstein2006generative}. Criminology has pursued this approach computationally, using agent-based simulations to test whether posited mechanisms suffice to generate observed patterns of offending  \citep{groff2007simulation,malleson2010crime,birks2012generative,birks2025agent}, while a parallel mathematical literature models the spatial concentration and network connections of crime \citep{short2008statistical,d2015statistical,klymentiev2025homophily,palma2025digital} and treats its recurrence as a phenomenon demanding explanation in its own right \citep{perc2013understanding}. We are not aware, however, of a previous attempt to formalise the dynamics of criminal propensity itself, as distinct from the spatial and situational distribution of crime events.

Evolutionary criminology, the study of crime-related mechanisms that have evolved naturally in humans \citep{svingen,svingen2025}, exists inside the umbrella of biosocial criminology. An evolutionary causal model aims to explain the emergence of crime or criminal propensity in human society through suitable combinations of both neurobiological and socio-psychological factors. The present paper develops a novel mathematical model of this kind, capturing the dynamics of individual criminal propensities in an agent-based system, to connect population-wide phenomena such as persistent or fluctuating crime to individual-level causes such as peer influence and environmental pressure. Our approach is to mathematise the Retribution and Reciprocity Model (RRM) developed by \cite{svingen}, which explores crime causation through the lens of cooperation (described in Section \ref{sec:RRM}). The mathematical model will extend the formalism of bounded confidence opinion dynamics (outlined in Section \ref{sec:OD}) to illuminate causal mechanisms of crime, allowing qualitative hypotheses of criminological theory to be interrogated through rigorous mathematical predictions. More broadly, the approach demonstrates a route by which verbal criminological theories that posit interactional mechanisms can be given formal expression and analytical treatment, with RRM serving here as the demonstration case. 

We present the mathematical RRM (mRRM) in Section \ref{sec:model}, including novel mechanisms of individual attractors and third-party witnesses. We prove that propensities remain bounded, establish sufficient conditions for convergence and find necessary conditions for persistently oscillatory behaviour, which is rare in bounded confidence models (Section \ref{sec:theorems}). Computational simulations reveal that persistent oscillations are the dominant mode of behaviour, and polarisation to the extremes tends to occur when agents share a neutral perception of the environment (Section \ref{sec:simulations}). We discuss the mathematical results through a criminological lens, particularly the propensity oscillations representing sustained cycles of offending tendencies, and the propensity polarisation which represents a population splitting into rule-abiding and rule-breaking oppositions (Section \ref{sec:discussion}). Conclusions are drawn together with an outline of future directions for mathematical formalisations of criminological theory (Section \ref{sec:conclusion}).

\section{Theoretical background} 
\subsection{Retribution and Reciprocity Model (RRM)} \label{sec:RRM}
An example of an evolutionary theory of crime is the Retribution and Reciprocity Model (RRM) developed by \cite{svingen}, exploring causal mechanisms through the lens of cooperation. Two important areas of cooperation are considered: retribution and reciprocity. \textit{Positive reciprocity} is defined to be the tendency to respond positively to positive actions (even if this incurs a significant cost to oneself). In contrast, \textit{negative reciprocity} is the tendency to respond negatively (for example with aggression) to a hostile action. \textit{Retribution} concerns one's response to events that do not directly involve oneself; it is the tendency to punish others who act against what is socially acceptable (even if enacting this punishment hurts oneself). These factors are combined with one's perception of the environment (PoE), which measures how well one believes they are generally treated by those around them, to determine to a large extent how likely one is to commit a crime the next time an opportunity arises (see Fig.~\ref{fig:rrm}). Experimental tests and the neurobiological basis for these causal mechanisms are explored by \cite{svingen}, who also notes that while cooperative tendencies in humans have evolved to aid the species' survival, there will always be individual and environmental variations that allow the persistence of socially detrimental behaviour -- or crime. 

RRM explains how crime emerges from interactions of everyday life -- some of them benign, some of them hostile or extreme -- through the way individuals are disposed to reciprocate and to punish. The model is inherently dynamic and interactional: criminal propensity is a quantity revised through encounters with others, and revised again by those who merely witness such encounters.  These features make RRM particularly suitable to a mathematisation within the formalism of dynamical systems. Other than the emergence of criminal propensity, another proposition of RRM raises interesting comparisons with nonlinear dynamics. As circumstances improve for a given society, free riders emerge who take advantage of others; as the situation worsens (partially due to free riders taking from the system), the majority punish the free riders and enough of them become cooperative again, resulting in the situation improving. This cycle keeps repeating, creating an oscillatory pattern. We note at the outset that the oscillations reported below are individual-level and arise from the sequence of interactions against a fixed environmental perception; they are not offered as a reproduction of this population-level cycle, and Section \ref{sec:conclusion} identifies the extension required to model it directly.

\begin{figure}[t!]
    \centering
    \includegraphics[width=0.55\linewidth]{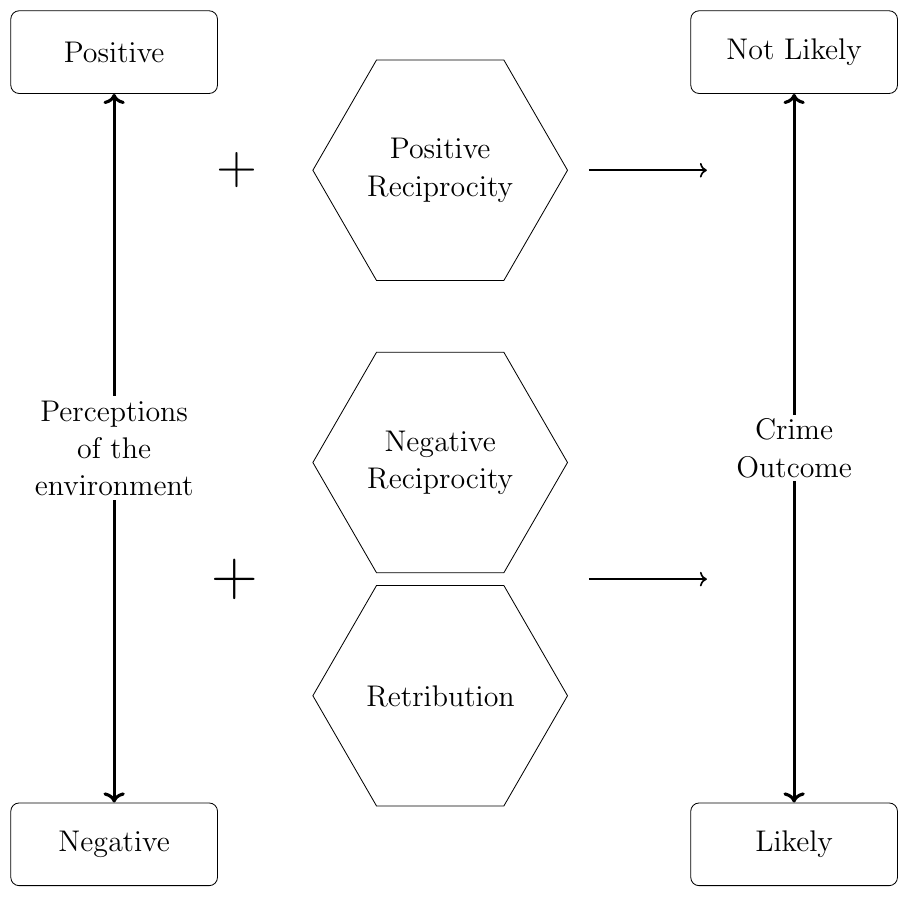}
    \caption{~The general framework of RRM, recreated from \cite{svingen}.}
    \label{fig:rrm}
\end{figure}

Formalisation of RRM is the natural next step in its development, since it forces the postulated mechanisms to be stated precisely enough for testing and, in principle, falsification. A mathematical model can access regions of behaviour that no survey or experiment could practically reach, and above all, it lets population-level patterns emerge from individual-level mechanisms (instead of assuming patterns at the outset). This is the standard of explanation that generative, agent-based social science sets for itself \citep{epstein2006generative}; it motivates the mathematical model we develop in this paper. The value of the exercise is visible in the results that follow. RRM in its verbal form does not predict that a population which agrees about its environment will divide only when that agreement is neutral, and has no means of generating such a prediction; yet the prediction follows from mechanisms the theory already posits, once those mechanisms are stated precisely enough to be iterated.

\subsection{Bounded confidence opinion dynamics}\label{sec:OD}
Models of opinion dynamics seek to explain the mechanistic underpinnings of evolutions in individual attitudes and behaviours, often revealing emergent collective phenomena such as consensus and clustering \citep{bernardo2024bounded}. A paradigmatic modelling framework is bounded confidence opinion dynamics, first developed at the turn of this century \citep{deffuant2000mixing,hegselmann2002opinion}. In a bounded confidence model, agents interact, compare opinions and update them to align with each other, provided the opinion difference is less than some threshold called the confidence bound. Individual opinions tend to stabilise over time, converging to either a common fixed point (representing the formation of a consensus opinion) or a number of fixed points (fragmentation into opinion clusters) \citep{lorenz2005stabilization}. Rarer forms of collective behaviour such as opinion oscillations have been found to emerge from some models where the confidence bounds are not prescribed constants but adapted to each agent and evolve over time \citep{stokes}. Persistent oscillations can also arise in stochastic systems \citep{sampson2025oscillatory}, but this paper pursues explanatory mechanisms for oscillations emerging from deterministic update rules.

The principal socio-psychological assumption behind bounded confidence models is homophily: individuals modulate their opinions to align with likeminded others \citep{mcpherson2001birds}. In the model introduced by \cite{deffuant2000mixing}, agent $i$ modifies their opinion following an interaction with agent $j$ if and only if there is already a sufficient level of agreement between them:
\begin{align}
v_i(t+1) = v_i(t)+ \delta (v_j(t)-v_i(t)) \quad \text{if } |v_j(t)-v_i(t)| \le \varepsilon_i,
\end{align}
where $v_i$ denotes the opinion of agent $i$, $\varepsilon_i$ the agent's confidence bound, and $\delta$ controls the amount of opinion modulation. Qualitatively, the dynamics of this model are almost entirely dependent on the confidence bounds \(\varepsilon_i\); consensus is more likely given wider bounds, while narrower bounds tend to yield fragmented clusters. In later developments on the model, the role of extremists in social networks is investigated, revealing that population-wide shifts to radical views can be driven by highly connected fringe individuals \citep{amblard2004role,weisbuch2005persuasion}. The asynchronous nature of opinion updates supports adaptation to a model of evolving criminal propensities through randomised social interactions. Moreover, collective behaviours such as consensus, polarisation and radicalisation, observed in models of the Deffuant type, are some of the key desired outcomes from a model of criminal propensity dynamics. For these reasons, the mathematical model developed here is an extension upon the Deffuant framework. The framework is expanded to incorporate a novel mechanism that enables oscillatory behaviour: agents can align their individual propensities either to extreme values or to their own perceptions of the environment. The resulting behaviour represents persistently fluctuating criminal propensities. 

\section{Mathematical model} \label{sec:model}

The mathematical RRM (mRRM) is built on two principles introduced by \cite{svingen}. First, agents evolve their criminal propensities over time due to discrete events where either social interactions or environmental factors influence them to change. Second, each agent's reciprocal and retributive tendencies control their responses to social interactions, while their perceptions of the environment determine how their criminal propensities change in the absence of strong interactions. 

We consider a system of \(N\) agents evolving their criminal propensities $C_i(t) \in [-1,1]$, where $i$ labels the agents, $t$ labels discrete time steps, $C_i = -1$ denotes the highest likelihood to commit a crime and $1$ the lowest (with $0$ being the neutral position). At each  $t$, a random even number $n(t) > 0$ of agents randomly form $n(t)/2$ interacting pairs, and the $N-n(t) \ge 0$ non-interacting agents are randomly assigned interactions to witness. There is no restriction on the number of witnesses that can be assigned to an interaction. When agents \(i\) and \(j\) interact, the totality of their reciprocal dynamics is captured in one time step, and any retributive dynamics of all witnesses to that interaction are captured in the same time step. For example, if agent \(i\) pushes agent \(j\) and \(j\) pushes \(i\) in retaliation while agent $k$ bears witness, then any resulting changes to their criminal propensities are captured in a single time step, rather than three sequential ones. Each agent is characterised by positive reciprocal, negative reciprocal, and retributive tendencies, denoted by $r_i^+, r_i^-,$ and $r_i^e$ respectively. These parameters take values in  $[0,1]$ and are assumed constant over time, but their variations among agents will be key to the system's dynamics. The boundedness of $C_i$ between $-1$ and $1$ is not imposed by boundary conditions, but a consequence of the dynamical update rules, as we will prove in Section \ref{sec:theorems}.

Each agent carries a constant perception of the environment (PoE), denoted by \(P_i \in[-1,1]\). The constancy of $P_i$ is the model's principal simplification relative to RRM,
in which perceptions are themselves revisable; it is adopted here for analytical tractability and its relaxation is taken up in Section 6. When agents $i$ and $j$ interact, we define the weight of their interaction as
\begin{align}
w_{ij} = \pm \frac{r_i^\pm}{4}  \left( 1 - |C_i| \right) \abs{C_j - P_i} ,
\label{eq:wformula}
\end{align}
where either the $+$ or the $-$ sign is used depending on whether the interaction is positively reciprocal or negatively reciprocal. Whenever $C_i$ and $C_j$ take values in $[-1,1]$, we have $w_{ij} \in [-1/2,1/2]$. Note that $w_{ij}$ is not identical to $w_{ji}$ but $\text{sgn} (w_{ij}) = \text{sgn} (w_{ji})$,
because an interaction is assumed to be either positively reciprocal or negatively reciprocal for both parties. This \(w_{ij}\) variable controls how much an agent's criminal propensity changes from its current value $C_i$ to its updated value $C_i^+$ due to a given interaction:
\begin{align}
C_i^+ = \begin{cases} C_i + w_{ij} \abs{C_j - C_i} & \text{if } r_i^\pm |C_j - P_i| > |P_i|, \\ C_i + (1-|P_i|)(1-|C_i|) (P_i - C_i) & \text{otherwise.} \end{cases} \label{eq:dynamics}
\end{align}
Swapping the labels \(i\) and \(j\) yields the way agent $j$ updates $C_j$. Equations \eqref{eq:wformula} and \eqref{eq:dynamics} capture the reciprocal dynamics of RRM. 

The top line of equation \eqref{eq:dynamics} is actualised for agent $i$ if the ``shock'' of an interaction, \(|C_j-P_i|\), is sufficiently large compared to the magnitude of the agent's PoE. The parameters $r_i^\pm$ are part of the threshold condition since agents with different reciprocal tendencies should have different thresholds for allowing pairwise influence to move their criminal propensity. This threshold is harder to meet for agents with more extreme PoE values, following the established finding that extremity is a dimension of attitude strength \citep{petty2023attitude} and that stronger attitudes resist revision \citep{pomerantz1995attitude}. The same pattern appears in crime-relevant populations, where aggressive boys' perceptions of an encounter track their prior expectations rather than the encounter itself, and persist even once the setting turns cooperative \citep{lochman1998distorted,tuente2019hostile}. The second line of \eqref{eq:dynamics} implements the RRM proposition that agents try to align themselves with their PoE in the absence of sufficiently strong social interactions. The \(P_i-C_i\) factor guarantees a ``natural drift'' towards \(P_i\), with the PoE acting as a ``natural propensity''.
When the PoE is extreme (\(P_i\) close to $\pm 1$), or when the propensity is extreme ($C_i$ close to $\pm 1$), the propensity drift towards PoE will be small, following the principle that extreme agents are stubborn.

Equation \eqref{eq:dynamics} is connected to Deffuant dynamics \citep{deffuant2000mixing}. In the case $\text{sgn} (w_{ij}) = \text{sgn} (C_j - C_i)$, the top line of \eqref{eq:dynamics} reduces to the Deffuant dynamics for agent $i$, $C_i^+ = C_i + |w_{ij}| (C_j - C_i)$ with $|w_{ij}| \in [0,1/2]$, but agent $j$ follows the ``anti-Deffuant'' dynamics, $C_j^+ = C_j - |w_{ji}| (C_i - C_j)$ with $|w_{ji}| \in [0,1/2]$. This means agents $i$ and $j$ both move in the direction of $C_j$ relative to $C_i$. Similarly, if $\text{sgn} (w_{ij}) = \text{sgn} (C_i - C_j)$, then agent $i$ follows anti-Deffuant dynamics and $j$ follows Deffuant dynamics, so that both agents move in the direction of $C_i$ relative to $C_j$. These dynamics reflect the RRM assumption that a positive (negative) interaction causes a positive (negative) change in both parties. 

We now introduce equations analogous to \eqref{eq:wformula} and \eqref{eq:dynamics}, to capture the retributive dynamics of witnesses. When an agent $i$ witnesses an interaction between agents $j$ and $k$, the perceived weight of the interaction, which controls the amount of movement in the witness' criminal propensity, is defined as 
\begin{equation}
    w_{i,jk}=\pm\frac{1}{4}r_i^e (1-|C_i|)\frac{|C_j-P_i|+|C_k-P_i|}{2},
    \label{eq:retw}
\end{equation}
where the sign indicates the positivity or negativity of the reciprocal interaction, so that \(\text{sgn}(w_{i,jk}) = \text{sgn}(w_{jk}) = \text{sgn}(w_{kj})\). Whenever $C_i$, $C_j$ and $C_k$ take values in $[-1,1]$, we have $w_{i,jk} \in [-1/2,1/2]$.
Then agent $i$'s criminal propensity moves from $C_i$ to $C_i^+$ according to
\begin{equation}
    C_i^+=
    \begin{cases} 
        C_i + w_{i,jk}\displaystyle\frac{|C_j-C_i|+|C_k-C_i|+|C_k-C_j|}{3} & \text{if } r_i^e \displaystyle\frac{|C_j-P_i|+|C_k-P_i|}{2}>|P_i|, \\
        C_i + (1-|P_i|)(1-|C_i|)(P_i-C_i) & \text{otherwise}.
    \end{cases}
    \label{eq:model2}
\end{equation}
The natural drift in \eqref{eq:model2} is identical to that in \eqref{eq:dynamics}. As for the witness dynamics following sufficiently strong interactions, factors such as $|C_j - P_i|$ in \eqref{eq:wformula} and $|C_j - C_i|$ in \eqref{eq:dynamics} are replaced by mean averages involving pairwise differences between the witness and the interacting agents. Equations \eqref{eq:wformula}--\eqref{eq:model2} constitute the full mRRM.

\section{Analytical results} \label{sec:theorems}

This section analytically investigates all regimes of model behaviour, namely, convergence to an extreme, convergence to a PoE, and persistent oscillations. We prove that all agents in the model have criminal propensities bounded between $-1$ and $1$ (Theorem \ref{thm:bounded}), persistent strongly positive (negative) interactions cause an agent's propensity to converge to $1$ ($-1$) (Theorem \ref{thm:converge1}), persistent weak interactions cause an agent's propensity to converge to their PoE (Theorem \ref{thm:converge2}), and find necessary conditions that enable persistently oscillatory propensities (Section \ref{sec:osc}).

\subsection{Boundedness and convergence} 

\begin{theorem}[Boundedness] \label{thm:bounded}
    Every agent $i$ in the model \eqref{eq:wformula}-- \eqref{eq:model2} has criminal propensity bounded between $-1$ and $1$ for all time,
    assuming $|C_i(0)| \le 1$ for all $i$.
\end{theorem}

\begin{proof}
    If suffices to show that if $|C_i| \le 1$, then $|C_i^+| \le 1$. Fix a time step $t$ and assume agent $i$ follows the natural drift, $C_i^+ = C_i + \left(1-|P_i|\right) \left(1-|C_i|\right) (P_i-C_i)$, which appears in both equations \eqref{eq:dynamics} and \eqref{eq:model2}. 
    This is equivalent to
    \begin{align}
    C_i^+ = \left(1-(1-|P_i|)(1-|C_i|)\right) C_i + (1-|P_i|)(1-|C_i|)P_i. \label{eq:proof1-0}
    \end{align}
    Whenever $|C_i| \le 1$, we have $(1-|P_i|)(1-|C_i|) \in [0,1]$ since $|P_i| \le 1$.
    Therefore $C_i^+$ is a weighted average of $C_i$ and $P_i$ with non-negative weights, and since both $C_i$ and $P_i$ take values in $[-1,1]$, so must $C_i^+$.
        
    To complete the proof, it remains for us to show that in both equations \eqref{eq:dynamics} and \eqref{eq:model2}, whenever $|C_i| \le 1$, the size of the change from $C_i$ to $C_i^+$ cannot exceed the difference between $C_i$ and the nearest extreme, which is $1$ if $C_i \ge 0$ or $-1$ if $C_i < 0$. That is, we need to show
    \begin{equation}
        |C_i^+-C_i|\leq 1-|C_i|.
    \end{equation}
    To that end, it suffices to show that for all $r_i^\pm,r_i^e \in [0,1]$, and all $P_i,C_i,C_j,C_k \in [-1,1]$, we have
    \begin{align}
    \frac{1}{4}r_i^\pm(1-|C_i|)|C_j-P_i|\, |C_j - C_i| &\le 1 - |C_i|, \label{eq:proof1-1} \\
    \frac{1}{4}r_i^e (1-|C_i|)\frac{|C_j-P_i|+|C_k-P_i|}{2}\,\frac{|C_j-C_i|+|C_k-C_i|+|C_k-C_j|}{3} &\le 1 - |C_i|. \label{eq:proof1-2}
    \end{align}
    Now, if $C_i = \pm1$, then $C_i^+ = \pm 1$. Assume $|C_i|<1$. Since $|C_j - P_i|$ and $|C_j - C_i|$ take values in $[0,2)$ for all $i,j$, we have: $r_i^\pm |C_j - P_i| \, |C_j - C_i| < 4$, which implies \eqref{eq:proof1-1}; $r_i^e (|C_j-P_i|+|C_k-P_i|) (|C_j-C_i|+|C_k-C_i|+|C_k-C_j|) < 24$, which implies \eqref{eq:proof1-2}.    
\end{proof}
The proof above also establishes the following corollary of Theorem \ref{thm:bounded}.
\begin{corollary} \label{cor}
    For every agent $i$ in the model \eqref{eq:wformula}-- \eqref{eq:model2}, the criminal propensity equals $\pm 1$ if and only if it equalled $\pm 1$ at the previous time step:
    \begin{align}
    C_i^+ = 1 \text{ if and only if } C_i = 1, \qquad C_i^+ = -1 \text{ if and only if } C_i = -1.
    \end{align}
    Moreover, if the agent follows the natural drift $C_i^+ = C_i + \left(1-|P_i|\right) \left(1-|C_i|\right) (P_i-C_i)$, then the updated criminal propensity takes value between the agent's current propensity and PoE:
    \begin{align}
    C_i^+ \in \left[\min\{C_i,P_i\}, \max\{C_i,P_i\}\right].
    \end{align}
\end{corollary}

The following two theorems establish sufficient conditions for $C_i$ to converge, given its initial value is not $\pm 1$. 
\begin{theorem}[Convergence to an extreme] \label{thm:converge1}
Consider some agent $i$ in the model \eqref{eq:wformula}--\eqref{eq:model2}. Suppose there exists time $T \ge 0$ at which $|C_i| \neq 1$ and for every $t \ge T$, agent $i$ either has an interaction with agent $j(t)$ such that $r_i^\pm |C_{j(t)}-P_i| > |P_i|$, or witnesses an interaction between agents $j(t)$ and $k(t)$ such that $r_i^e (|C_{j(t)}-P_i| + |C_{k(t)}-P_i|)/2 > |P_i|$. Suppose also there exists some $\delta > 0$ such that $|C_{j(t)} - P_i| \,|C_{j(t)} - C_i| \ge \delta$ and $(|C_{j(t)} - P_i| + |C_{k(t)}-P_i|)\,(|C_{j(t)} - C_i| + |C_{k(t)} - C_i| + |C_{k(t)} - C_{j(t)}|) \ge \delta$ for all $t \ge T$. 
    \begin{enumerate}[label=(\roman*)]
    \item If all interactions involving or witnessed by agent $i$ for $t \ge T$ are positive, then
    \begin{equation}
        \lim_{t\rightarrow\infty} C_i(t) = 1.
    \end{equation}
    \item If all interactions involving or witnessed by agent $i$ for $t \ge T$ are negative, then
    \begin{equation}
        \lim_{t\rightarrow\infty} C_i(t) = -1.
    \end{equation}
    \end{enumerate} 
\end{theorem}
\begin{proof} 
According to equations \eqref{eq:wformula}--\eqref{eq:model2}, if $r_i^\pm |C_{j(t)}-P_i| > |P_i|$ or $r_i^e (|C_{j(t)}-P_i| + |C_{k(t)}-P_i|)/2 > |P_i|$ for a positive interaction, then $C_i^+ \ge C_i$. Similarly, for a negative interaction, $C_i^+ \le C_i$. Therefore $C_i(t)$ is an increasing (decreasing) sequence for $t \ge T$ if all interactions are positive (negative), and since $C_i$ is bounded above and below by Theorem \ref{thm:bounded}, it must be convergent in either case. It remains to be shown that the limit as $t \to \infty$ must be $1$ in the first case and $-1$ in the second. 

Let $\lim_{t \to \infty} C_i(t) = C_i^\infty$. We will present a proof that $C_i^\infty = 1$ if all interactions are positive, and the argument for the opposite case will be identical and therefore omitted. Let $t_m$ with $m = 1,2,3,\ldots$ ($t_n$ with $n = 1,2,3,\ldots$) be the strictly increasing sequence of time steps at which agent $i$ interacts positively (witnesses positive interactions). Then the sequence of time steps $t \ge T$ is partitioned into $t_m$ and $t_n$, so at least one of these subsequences must be infinite. 

Suppose $t_m$ is infinite, then $C_i^\infty = \lim_{m \to \infty} C_i(t_m)$. Consider the sequence of agents, $j(t_m)$, with whom $i$ interacts at times $t_m$. If $j(t_m)$ does not converge as $m \to \infty$, then letting $t \to \infty$ in equations \eqref{eq:wformula} and \eqref{eq:dynamics} immediately yields $|C_i^\infty| = 1$. If $j(t_m)$ converges as $m \to \infty$, then there must exist some agent $j$ and some time $T'$ such that $j(t_m) = j$ for all $t_m \ge T'$, and so letting $t \to \infty$ in \eqref{eq:wformula} and \eqref{eq:dynamics} gives
\begin{align}
0 = \frac{1}{4} r_i^+ (1-|C_i^\infty|)|C_{j}-P_i|\, |C_{j}-C_i^\infty|.
\end{align}
Since $r_i^+$ must be positive to ensure $r_i^+ |C_{j(t)}-P_i| > |P_i|$ for all $t \ge T$, and since $|C_{j} - P_i| \,|C_{j} - C_i^\infty| \ge \delta$ for some positive constant $\delta$, we must have $|C_i^\infty| = 1$.

Now suppose $t_m$ is finite, so $t_n$ must be infinite. Then $C_i^\infty = \lim_{n \to \infty} C_i(t_n)$. Consider the sequences of agents, $j(t_n)$ and $k(t_n)$, whose interactions are witnessed by agent $i$ at time steps $t_n$. If either $j(t_n)$ or $k(t_n)$ does not converge as $n \to \infty$, then letting $t \to \infty$ in equations \eqref{eq:retw} and \eqref{eq:model2} immediately yields $|C_i^\infty| = 1$. If both $j(t_n)$ and $k(t_n)$ converge as $n \to \infty$, then there must exist some agents $j,k$ and some time $T'$ such that $j(t_n) = j$ and $k(t_n) = k$ for all $t_n \ge T'$, and so letting $t \to \infty$ in \eqref{eq:retw} and \eqref{eq:model2} gives
\begin{align}
0 = \frac{1}{4}r_i^e (1-|C_i^\infty|) \frac{|C_j-P_i|+|C_k-P_i|}{2} \, \frac{|C_j-C_i^\infty|+|C_k-C_i^\infty|+|C_k-C_j|}{3} .
\end{align}
Since $r_i^e$ must be positive to ensure $r_i^e (|C_{j(t)}-P_i| + |C_{k(t)}-P_i|)/2 > |P_i|$ for all $t \ge T$, and since $(|C_{j} - P_i| + |C_{k}-P_i|)\,(|C_{j} - C_i^\infty| + |C_{k} - C_i^\infty| + |C_{k} - C_{j}|) \ge \delta$ for some positive constant $\delta$, we must have $|C_i^\infty| = 1$. 

In each case, we have shown $|C_i^\infty| = 1$. The fact that $C_i(t)$ is increasing due to positive interactions, with $C_i(T) > -1$, then implies $C_i^\infty = 1$. 
\end{proof}
\begin{theorem}[Convergence to PoE] \label{thm:converge2}
    Consider some agent $i$ in the model \eqref{eq:wformula}--\eqref{eq:model2} with $|P_i| \neq 1$. Suppose there exists time $T \ge 0$ at which $|C_i| \neq 1$ and for every $t \ge T$, agent $i$ either has an interaction with agent $j(t)$ such that $r_i^\pm |C_{j(t)}-P_i| \le |P_i|$, or witnesses an interaction between agents $j(t)$ and $k(t)$ such that $r_i^e (|C_{j(t)}-P_i| + |C_{k(t)}-P_i|)/2 \le |P_i|$. Then, 
    \begin{equation}
        \lim_{t\rightarrow\infty} C_i(t) = P_i.
    \end{equation}
\end{theorem}
\begin{proof}
By assumption, $C_i$ follows the natural drift,
\begin{align}
C_i^+ = C_i + (1-|P_i|)(1-|C_i|)(P_i-C_i), \label{eq:drift}
\end{align}
for every $t \ge T$. If $C_i(T) = P_i$, we are done, since $P_i$ is a fixed point of equation \eqref{eq:drift}. Otherwise, $C_i^+ \le C_i$ if $C_i > P_i$ and $C_i^+ \ge C_i$ if $C_i < P_i$. Since $C_i^+$ is bounded between $C_i$ and $P_i$ by Corollary \ref{cor}, $C_i(t)$ for $t \ge T$ must be monotonic, and therefore convergent. Letting $t \to \infty$ in equation \eqref{eq:drift} yields $0 = (1-|P_i|)(1-|C_i^\infty|)(P_i-C_i^\infty)$. Since $|P_i| \neq 1$, we must have either $C_i^\infty = \pm 1 $ or $C_i^\infty = P_i$, and the fact that $C_i$ is bounded by $P_i \in (-1,1)$ enforces $C_i^\infty = P_i$.
\end{proof}

\subsection{A necessary condition for propensity oscillations}\label{sec:osc}

For oscillatory criminal propensities to arise, the top lines of equations \eqref{eq:dynamics} and \eqref{eq:model2} must activate persistently. At any given time step for any given agent $i$, the threshold condition in equation \eqref{eq:dynamics} can hold only if $r_i^\pm \max_{j} \lvert C_j-P_i\rvert > \lvert P_i\rvert$, which can hold only if at least one of the following is true:
\begin{align}
r_i^\pm (1-P_i) > |P_i|, \quad r_i^\pm (1+P_i) > |P_i|.  \label{eq:osc}
\end{align}
The first inequality in \eqref{eq:osc} can hold only if $r_i^\pm > 0$ and $-r_i^\pm/(1-r_i^\pm) < P_i < r_i^\pm/(1+r_i^\pm)$, while the second can hold only if $r_i^\pm > 0$ and $-r_i^\pm/(1+r_i^\pm) < P_i < r_i^\pm/(1-r_i^\pm) $. Overall, therefore, a necessary condition for oscillations to arise from  \eqref{eq:dynamics} is
\begin{align}
r_i^\pm > 0, \quad -r_i^\pm/(1-r_i^\pm) < P_i <r_i^\pm/(1-r_i^\pm), \label{eq:osc2}
\end{align}
where $r_i^\pm = 1$ yields an infinite range for $P_i$. 
Note that if $r_i^\pm > 1/2$, then $-r_i^\pm/(1-r_i^\pm) < -1$ and $r_i^\pm/(1-r_i^\pm) > 1$, so the second part of \eqref{eq:osc2} is automatically satisfied due to the parameter constraint $-1 \le P_i \le 1$. Overall, therefore, this necessary condition for persistently oscillatory reciprocal dynamics can be expressed as
\begin{align}
P_i \in\begin{cases}
   \varnothing, &\text{if } r_i^{\pm}=0,\\ \left( -\frac{r_i^{\pm}}{1-r_i^{\pm}},\frac{r_i^{\pm}}{1-r_i^{\pm}}\right), & \text{if } 0<r_i^{\pm}\le \frac{1}{2},\\ [-1,1], & \text{if } \frac{1}{2} < r_i^{\pm} \leq 1.
\end{cases} \label{eq:cond-osc}
\end{align}
Similarly, for propensity oscillations to arise from \eqref{eq:model2}, a necessary condition is identical to \eqref{eq:cond-osc} with $r_i^\pm$ replaced by $r_i^e$. 

If retribution and reciprocity parameter values are drawn from the standard normal distribution centred at \(1/2\) truncated at \(0\) and \(1\) (see Section \ref{sec:simulations}), then any randomly generated population is expected to have more than half of the agents meeting condition \eqref{eq:cond-osc}. These agents have the capacity to exhibit oscillatory propensities due to reciprocal dynamics, though they are not guaranteed to do so; while the remaining portion of the population do not have the capacity to oscillate. The same is true for oscillations due to retributive dynamics. Since the parameters $r_i^\pm$ and $r_i^e$ are chosen independently, any population randomly generated in this way is expected to have fewer than a quarter of the agents being unable to exhibit oscillations of any kind. The model setup therefore enables propensity oscillations to be the dominant mode of behaviour for most agents, as confirmed by computational simulations of the model, presented in the next section.

\section{Simulations}\label{sec:simulations}

We simulate the model dynamics of criminal propensities using Python, to generate a broad spectrum of behaviour beyond what is captured by the analytical results in Section \ref{sec:theorems}. These simulations will reveal the capabilities and limitations of the mathematical framework while offering insights into the emergence of crime under RRM.

For each simulation, we initialise \(100\) agents at time $t=0$ and track their criminal propensities for up to \(10,000\) time steps. Each agent's positive and negative reciprocity parameters are held equal ($r_i = r_i^\pm$). This assumed equality reduces the parameter count and therefore the number of simulations required to investigate emergent patterns, while still allowing the full range of qualitative behaviour to emerge. The interaction events are randomly assigned either a positive or negative nature at each timestep following independent Bernoulli(\(0.5\)) distributions, representing equal likelihood of the two types of interaction. Parameter values of \(r_i^\pm, r_i^e\) and $P_i$, together with initial values of the criminal propensities, are generated independently from normal distributions truncated at the appropriate values. The independence and normality assumptions align with experimental results discussed in \cite{svingen}.
Through these simulations, we find general relationships between retribution, reciprocity, perception of the environment (PoE), and criminal propensity values. 
We closely examine populations where everyone's view of the environment equals some universal PoE value, $P$, that is close to the neutral value 0. Parameter regimes that give rise to consensus, polarisation, and propensity oscillations are investigated.

\subsection{Results}

While bounded confidence models generally predict convergent dynamics, the dominant mode of behaviour here is persistent oscillations of individual propensities, as seen in Fig.~\ref{fig:standard}. The same figure illustrates that when an agent experiences or witnesses sufficiently strong interactions, their criminal propensity can be pulled away from their PoE to stabilise near an extreme value ($\pm 1$). When all agents view the environment identically (the PoE value is set to some universal \(P_i=P\) for all \(i\)), all criminal propensities converge to \(P\) if $P$ is sufficiently far from 0 (see Fig. \ref{fig:UnivP}), and if \(P \approx 0\), polarisation can happen, where the majority of agents split into two factions and converge to $\pm 1$ (see Fig. \ref{fig:UnivP0}). Uniformity of perception promotes consensus unless that perception is neutral, where it produces maximal polarisation, which inverts the usual expectation that neutrality yields agreement. A further exploration of this result is presented in Section \ref{sec:grids}.

\begin{figure}[t!]
    \centering
    \includegraphics[width=.6\linewidth]{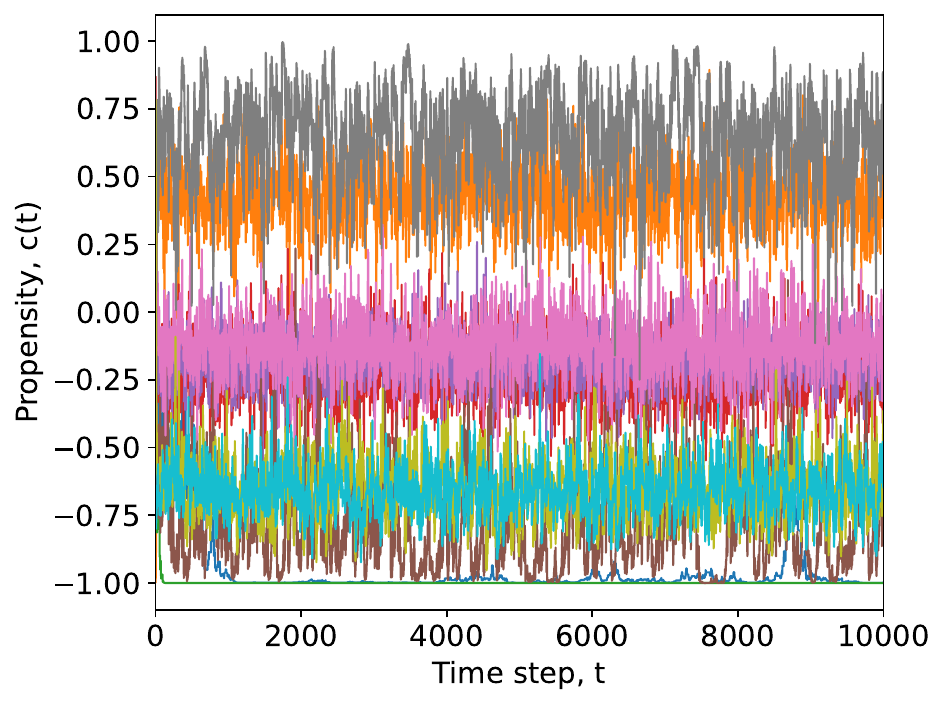}
    \caption{~Example simulation via Python of the model with 100 agents over 10,000 time steps, where only 10 agents are visualised for clarity. Parameters are fixed upon initialisation at randomised values. The common behaviour of criminal propensities oscillating around PoE values over time  is illustrated, as is the rarer phenomenon of criminal propensities diverging from PoE and stabilising around extreme values (in this case, $-1$).}
    \label{fig:standard}
\end{figure}

\begin{figure}[t!]
     \centering
     \begin{subfigure}{0.49\textwidth}
         \centering
         \includegraphics[width=\linewidth]{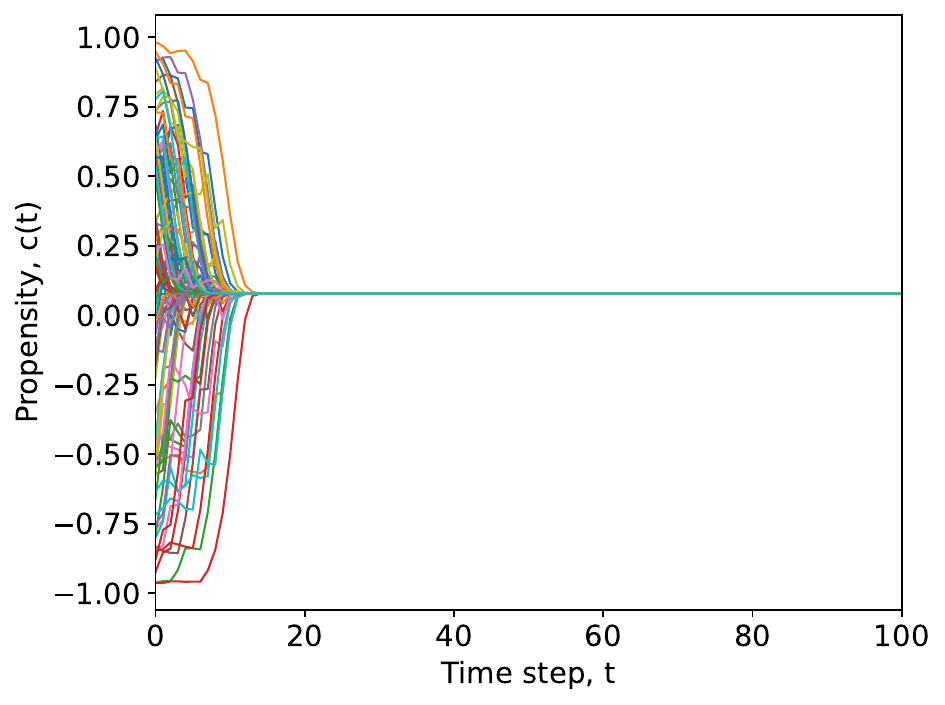}
         \caption{~\(P=\overline{{C_i}(0)}\), the initial mean average of ${C_i}$.}
         \label{fig:UnivP}
     \end{subfigure}
     \hfill
     \begin{subfigure}{0.49\textwidth}
         \centering
         \includegraphics[width=\linewidth]{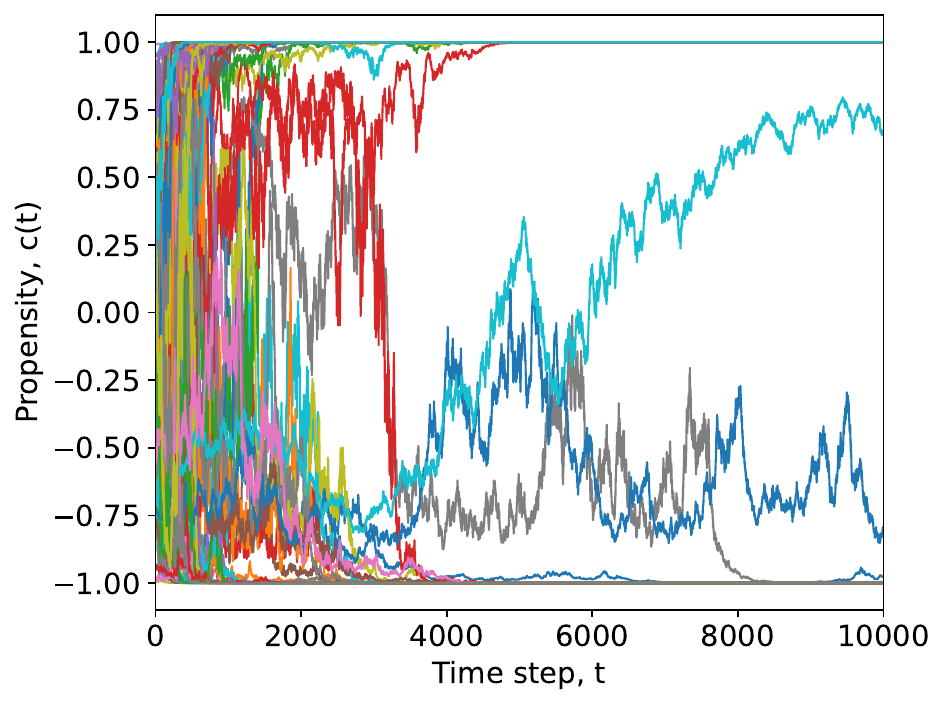}
         \caption{~\(P=0\).}
         \label{fig:UnivP0}
     \end{subfigure}
        \caption{~Example simulations via Python of the model with 100 agents over (a) 100 time steps, and (b) 10,000 time steps. Parameters are fixed upon initialisation at randomised values, except for the universal PoE value, \(P_i=P\) for all \(i\).}
        \label{fig:UnivPs}
\end{figure}

It is important to study the roles played by the  retribution (\(r^e_i\)) and reciprocity (\(r_i^\pm\)) parameters in determining the dynamics. When all agents' positive and negative reciprocities are $r=0$, every interaction results in the participants moving their criminal propensity closer to their PoE (see equation \ref{eq:dynamics}); similarly, when all retribution parameters are set to $r^e = 0$, every witness moves their propensity closer to their PoE after each interaction (see equation \ref{eq:model2}). In both cases, a significant portion of criminal propensities converge to the individual PoE values over time, but the majority still oscillate around the PoE values without converging (see Fig.~\ref{fig:r0re0}). 

\begin{figure}[t!]
     \centering
     \begin{subfigure}{0.49\textwidth}
         \centering
         \includegraphics[width=\linewidth]{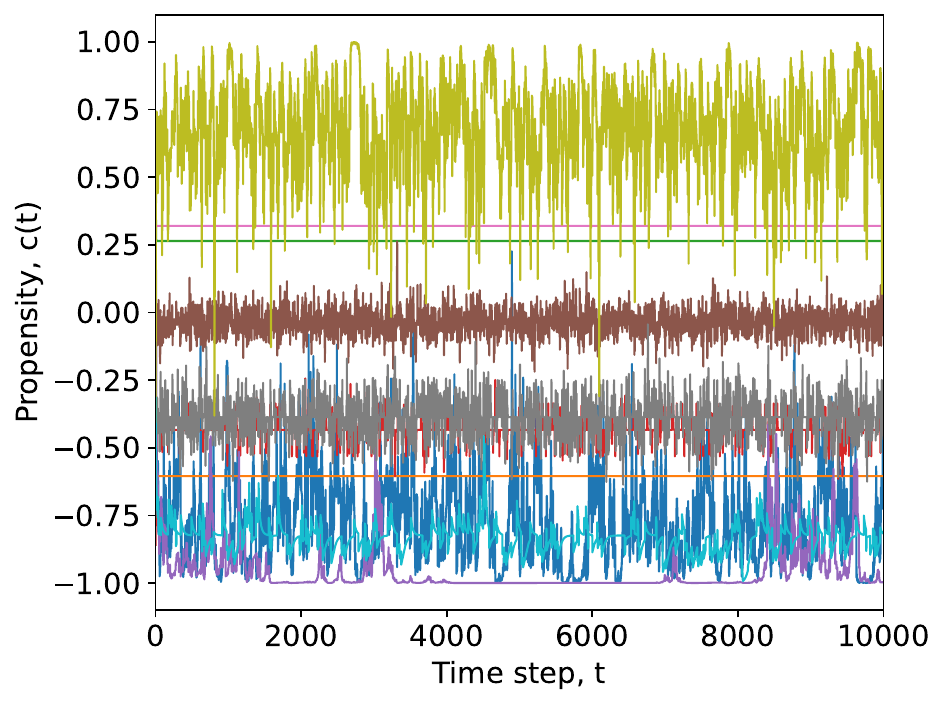}
         \caption{~\(r^e=0.\)}
         \label{fig:r0}
     \end{subfigure}
     \hfill
     \begin{subfigure}{0.49\textwidth}
         \centering
         \includegraphics[width=\linewidth]{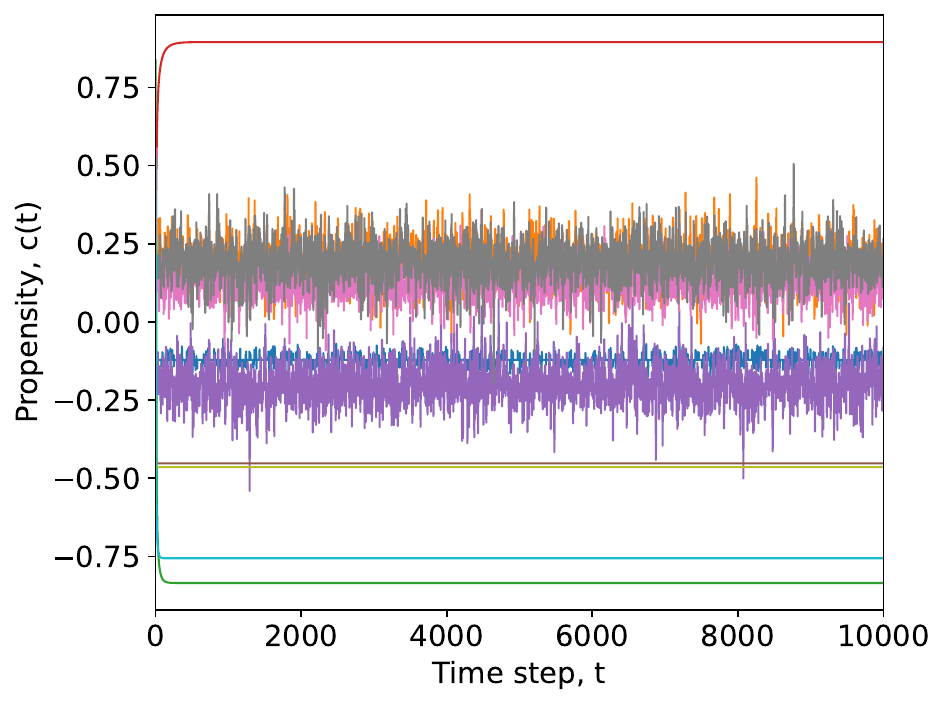}
         \caption{~\(r=0\).}
         \label{fig:re0}
     \end{subfigure}
        \caption{~Example simulation via Python of the model with 100 agents over 10,000 time steps, where only 10 agents are visualised for clarity. Parameters are fixed upon initialisation at randomised values, except for either (a) the universal retribution $r_i^e = r^e$, or (b) the universal reciprocity (both positive and negative) $r_i = r$.}
        \label{fig:r0re0}
\end{figure}

A robust feature of the simulations is that as retribution parameters increase, for example as a universal value $r_i^e = r^e$ for all $i$, the oscillations of criminal propensities tend to exhibit larger amplitudes (compare Fig.~\ref{fig:re01} and \ref{fig:re09} to Fig.~\ref{fig:r0}). 
The same result holds when increasing the reciprocity parameters (compare Fig.~\ref{fig:r01} and \ref{fig:r09} to Fig.~\ref{fig:r0}). In both cases, if individual variability is allowed (agents have individually allocated $r_i^e$ or $r_i^\pm$), the same results hold: larger retributive or reciprocal tendencies create larger oscillations in criminal propensities. If the retribution and reciprocity parameters are generally small, most agents will have their propensities converge to their PoE, and the agents whose propensities are near zero can undergo small oscillations (see Fig.~\ref{fig:r01re01}).
If retribution and  reciprocity are high (close to \(1\)), the population will likely undergo polarisation, with some propensities converging to $1$ and others to $-1$ (see Fig.~\ref{fig:r09re09}).

\begin{figure}[t!]
     \centering
     \begin{subfigure}{0.49\textwidth}
         \centering
         \includegraphics[width=\linewidth]{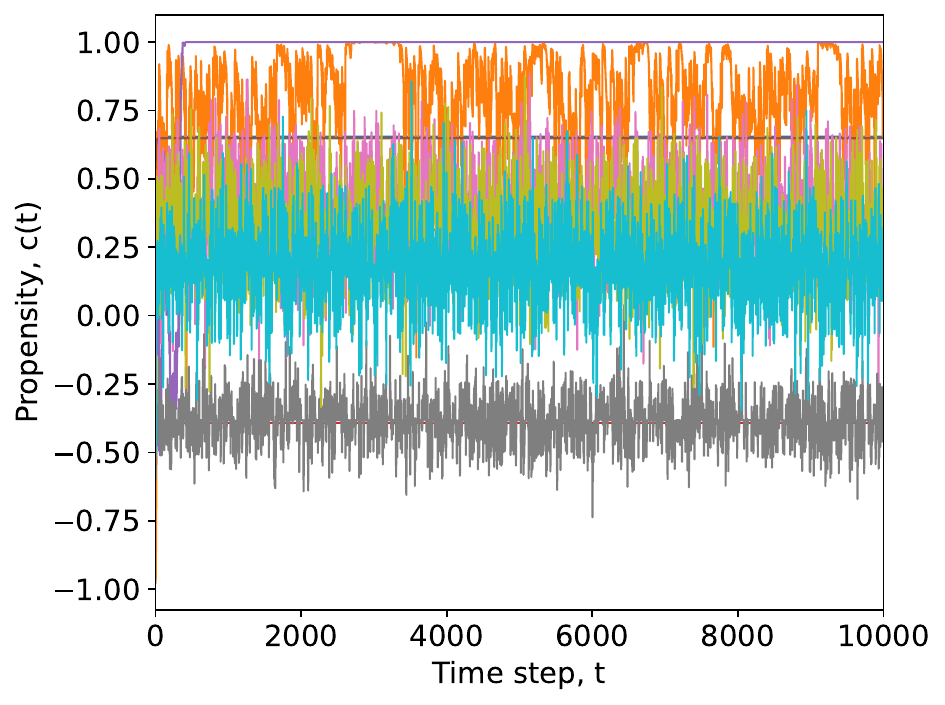}
         \caption{~\(r^e=0.1.\)}
         \vspace{10pt}
         \label{fig:re01}
     \end{subfigure}
     \hfill
     \begin{subfigure}{0.49\textwidth}
         \centering
         \includegraphics[width=\linewidth]{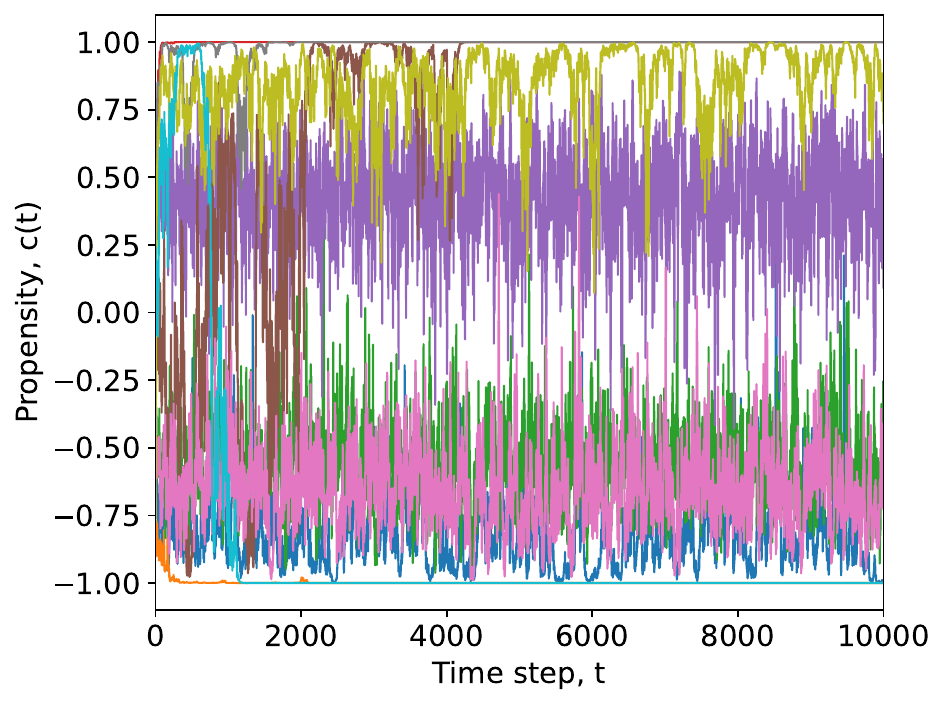}
         \caption{~\(r^e=0.9.\)}
         \vspace{10pt}
         \label{fig:re09}
     \end{subfigure}
     \begin{subfigure}{0.49\textwidth}
         \centering
         \includegraphics[width=\linewidth]{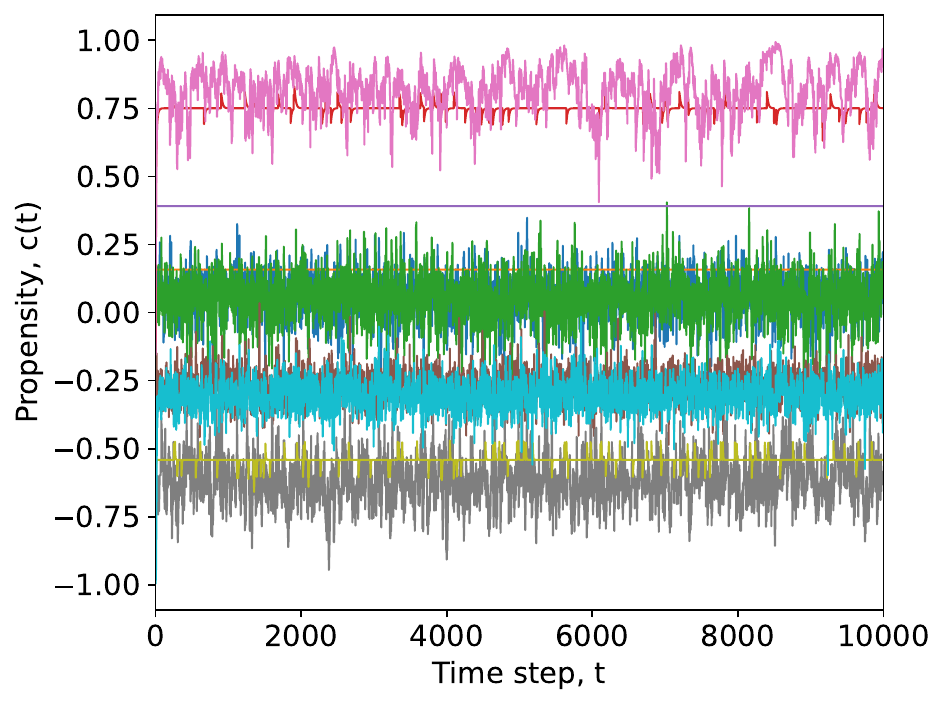}
         \caption{~\(r=0.1.\)}
         \vspace{10pt}
         \label{fig:r01}
     \end{subfigure}
     \hfill
     \begin{subfigure}{0.49\textwidth}
         \centering
         \includegraphics[width=\linewidth]{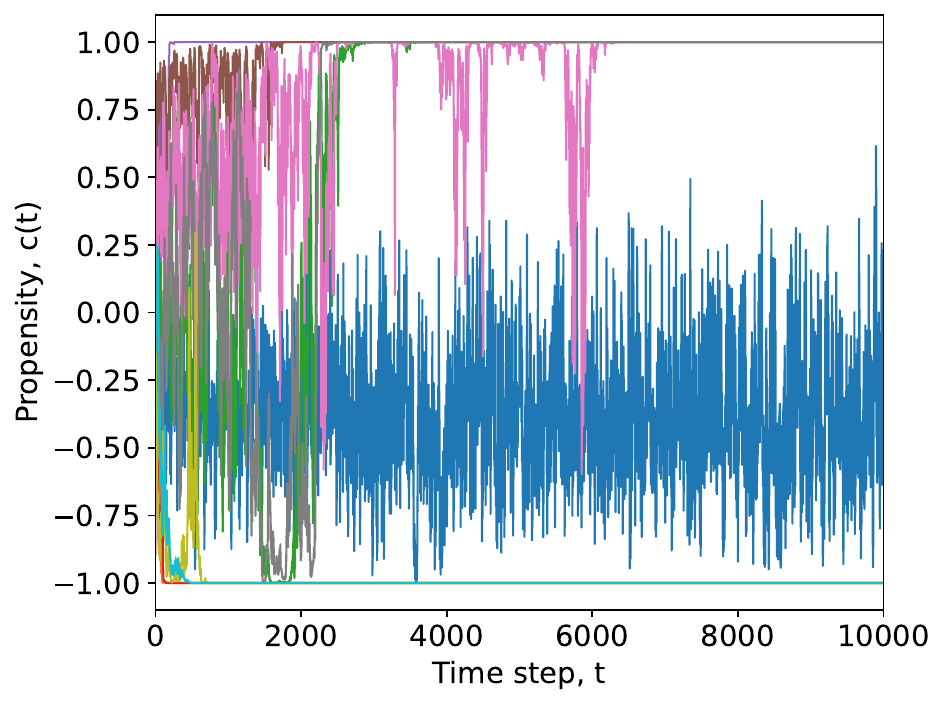}
         \caption{~\(r=0.9.\)}
         \vspace{10pt}
         \label{fig:r09}
     \end{subfigure}
     \begin{subfigure}{0.49\textwidth}
         \centering
         \includegraphics[width=\linewidth]{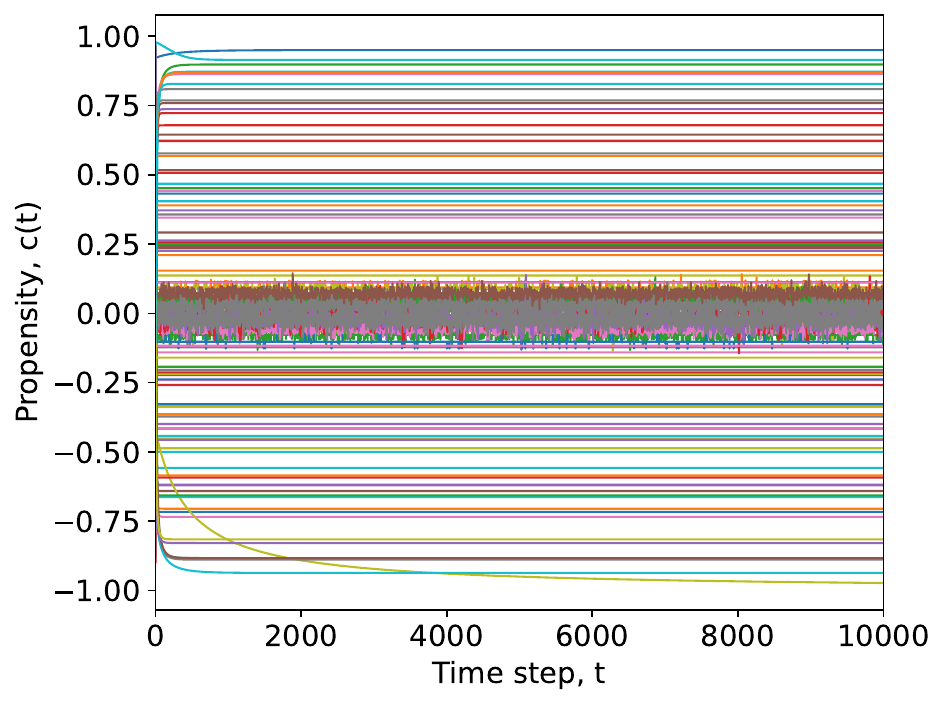}
         \caption{~\(r=r^e=0.1.\)}
         \label{fig:r01re01}
     \end{subfigure}
     \hfill
     \begin{subfigure}{0.49\textwidth}
         \centering
         \includegraphics[width=\linewidth]{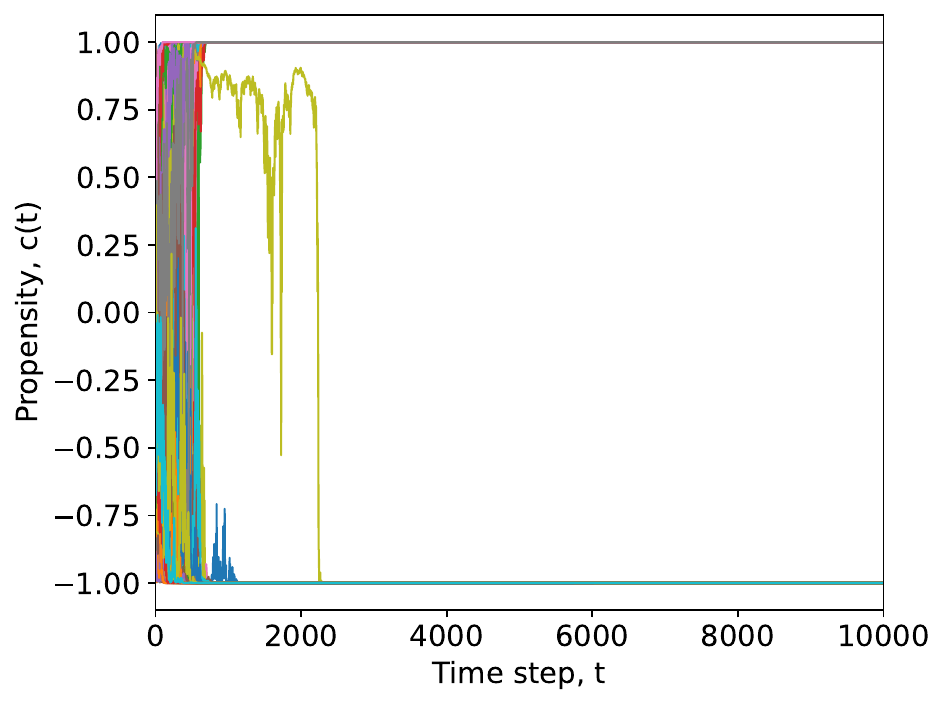}
         \caption{~\(r=r^e=0.9\).}
         \label{fig:r09re09}
     \end{subfigure}
        \caption{~Example simulation via Python of the model with 100 agents over 10,000 time steps, where only 10 agents are visualised for clarity. Parameters are fixed upon initialisation at randomised values, except for either (a,b) the universal retribution $r^e_i = r^e$, or (c,d) the universal reciprocity (both positive and negative) $r_i^\pm = r$, or (e,f) both.}
        \label{fig:rre}
\end{figure}

\subsection{Consensus and polarisation under near-neutral PoE}\label{sec:grids}

To examine the behaviour when the universal PoE is \(P\approx0\), for a given population size $N$, we set the universal retribution and reciprocity to be equal, \(r=r^e_i=r^\pm_i\), to reduce the parameter count, then fix \(C_i(0)\), the initial values of criminal propensities, fix the agent pairings that occur at each time step, and fix each pairing's interaction type at each time step.
We then vary \(r\) and \(P\) over a \(20\)-by-\(20\) grid of values. 
In short, out of the \(400\) simulations, the only varied settings are \(r\) and \(P\). 
The overall behaviour of the population can be consensus (all propensities converge to the same value), polarisation (a splitting of the population into factions that converge to $\pm 1$), or inconclusive (neither consensus nor polarisation over 10,000 time steps). Generally, low values of $P$ give rise to polarisation while large values give rise to consensus, and as the value of \(r\) increases, so does the critical PoE that separates the polarisation regime from the consensus one (see Fig. \ref{fig:grids}). Moreover, increasing the number of agents decreases the critical PoEs (compare Fig.~\ref{fig:grid10} to \ref{fig:grid100}).
For a given $r$ value, increasing \(P\) can cause the simulation outcome to switch between polarisation and consensus multiple times (see Fig.~\ref{fig:grid10} for $r =0.75$ and Fig.~\ref{fig:grid100} for $r =1$). This effect demonstrates that near the critical PoEs, the dynamics are very sensitive; \(P\) changing marginally can have large knock-on effects. 

\begin{figure}[t!]
     \centering
     \begin{subfigure}{0.49\textwidth}
         \centering
         \includegraphics[width=\linewidth]{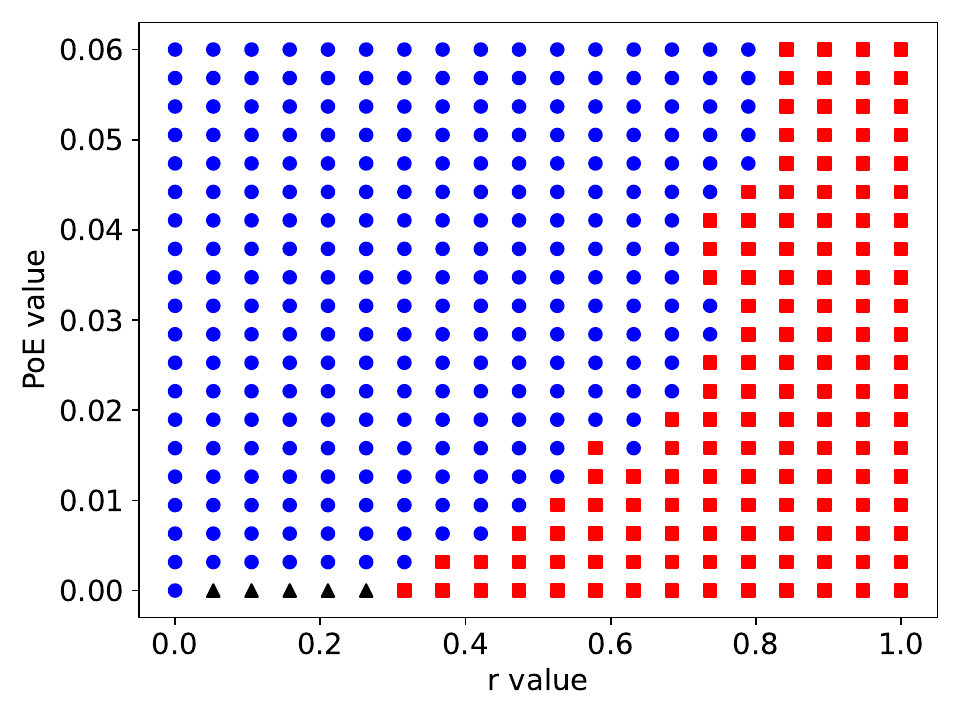}
         \caption{~\(N=10\).}
         \label{fig:grid10}
     \end{subfigure}
     \hfill
     \begin{subfigure}{0.49\textwidth}
         \centering
         \includegraphics[width=\linewidth]{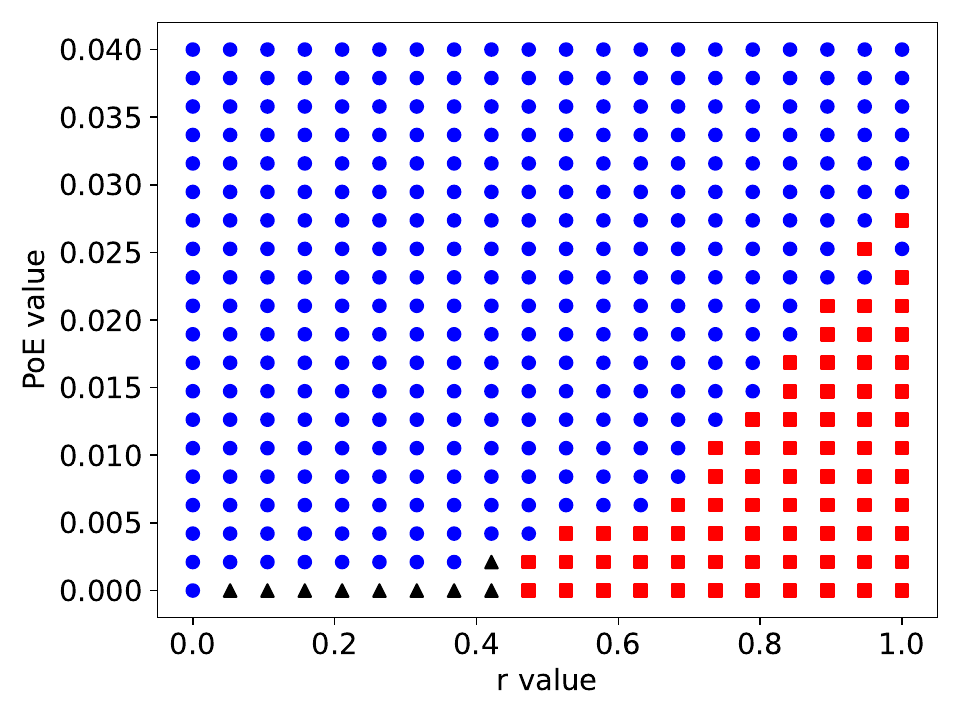}
         \caption{~\(N=100\).}
         \label{fig:grid100}
     \end{subfigure}
        \caption{~Simulations on a \(20\)-by-\(20\) grid of universal $r=r^e=r^\pm$ and \(P\) parameter values, with either (a) $10$ agents or (b) $100$ agents. A blue circle denotes that the simulation outcome is consensus, a red square denotes polarisation, and a black triangle means that the outcome is inconclusive as convergence does not occur within 10,000 time steps.}
        \label{fig:grids}
\end{figure}

\section{Discussion}\label{sec:discussion}

The analytical results in Section \ref{sec:theorems} and the simulations in Section \ref{sec:simulations} together constitute an overall picture of the behaviour of agents under mRRM. An agent's perception of the environment (PoE) serves as the their ``natural criminal propensity'', and the dynamical equations dictate that individual propensities are naturally attracted to the PoE values unless disrupted by sufficiently strong interaction events. The results show that most agents tend to have their propensities oscillate around their PoE, even if some agents abandon their PoE values permanently due to excessively strong interaction events. These quasi-stable fluctuations are reminiscent of, though not a reproduction of, the free-rider cycles theorised within RRM \citep{svingen}. There the oscillation is a population-level alternation between exploitation and punishment driven by a changing environment; here it is an individual-level alternation driven by the interaction sequence against a fixed perception of that environment. The correspondence is one of form rather than of mechanism, and Section \ref{sec:conclusion} identifies the extension that would close the gap. What the model does establish is the weaker but still substantive claim that RRM's mechanisms are sufficient to generate persistent non-convergence with no external driver, which is the property the free-rider account requires of them.

When all agents are assigned the same PoE value ($P$), all propensities converge to this universal PoE, provided \(P\) is sufficiently far away from \(0\). A non-neutral uniform PoE therefore promotes consensus. Where perceptions are shared, agents respond to the environment and to one another in the same way, and the population converges on the common perception; the exception is the near-neutral case, in which agents are more readily displaced towards the extremes. Closer analysis shows this to be a direct consequence of the threshold specification rather than an artefact of implementation. When $P \approx 0$, it is ``easy'' for agents to move their criminal propensities towards the extremes rather than towards their own PoE value, since any interactions that do not conform to the agent's PoE will be considered ``shocking'' (see equations \ref{eq:dynamics} and \ref{eq:model2}). The fact that agents converge to $\pm 1$, rather than oscillate near those values, corroborates Theorem \ref{thm:converge1}: consistently shocking events cause convergence to an extreme. In RRM's terms, an agent whose PoE is neutral holds no settled expectation of how they are generally treated, and so has no baseline against which an encounter could be judged unremarkable \citep{svingen}. The model's prediction is that a population in this condition is maximally exposed to whatever its members happen to do to one another. 

Notably, given a universal PoE $P \approx 0$, polarisation would be less likely if all agents have neutral criminal propensities ($C_i \approx 0$ for all $i$), because the impetus to move towards the extreme would be small ($|C_j - P_i| \approx 0$ for all $i$ and $j$). In such a scenario, a population undergoing random interactions would likely remain neutral. It is the neutral PoE that creates the necessary condition for individual radicalisation under extremist influences, and the influence of agents holding outlying propensities that converts that susceptibility into realised division. 

The observation that larger populations have smaller ranges of $P$ for which polarisation occurs, or more values of $P$ for which consensus occurs, may reflect a form of self-averaging: as the population size increases, the aggregate influence experienced by each agent becomes closer to its expected value, reducing the impact of fluctuations. In a small population, a few early extreme excursions can disproportionately shape the interaction environment experienced by others, increasing the likelihood that subsequent interactions are perceived as shocking and thereby triggering further movement towards the extremes. Although larger populations are more likely to contain individuals with outlying propensities, their ability to exacerbate polarisation seems suppressed by the more dominant stabilising effect of averaging over many interactions, according to the simulation results.

Investigating the effects of retribution ($r_i^e$) and reciprocity ($r_i^\pm$), we found that they mainly control propensity oscillations. For an explanation of this dependence, we look at equations \eqref{eq:dynamics} and \eqref{eq:model2}. When $r^e_i$ or  $r_i^\pm$ is large, agents are likely to be influenced by other individuals and move their propensities towards extreme values accordingly (rather than moving towards their own PoE values). Conceptually, if \(r_i^\pm\) is large then the agent is more reciprocal, less stubborn, and so allows for larger changes in their propensity. The same can be said for agents with large $r^e_i$: their large retributive tendencies allow for large oscillations. It follows that as either retribution or  reciprocity (or both) increases for an agent, they will likely exhibit more variations in criminal propensity; oscillations occur when there is a balance between the number of times the agent moves towards an extreme (influenced by another individual) and the number of times they move towards their PoE. An agent with an extreme PoE value (close to $\pm 1$) is likely to align their criminal propensity to that value over time, since both individual influence and environmental influence induce movement towards an extreme. Thus, the model predicts that extreme perceptions of the environment beget stubbornness in criminal propensity. 

When both retribution and reciprocity are high (close to 1), polarisation sometimes follows, though it is not guaranteed. This happens because the agents are so cooperative that they are likely to be moved by individual influence and by large amounts, forcing them towards the extreme propensity values. Since $\pm 1$ are sinks in the dynamical system, it is relatively easy for an agent to become radicalised over time but considerably harder to move back towards neutral once they have obtained extreme propensity values ($C_i \approx \pm 1$). 

Three limitations bear on the interpretation of these results. First, the valence of each interaction is drawn independently of the agents involved, so that criminal propensity does not itself shape the encounters an agent has. This is a defensible reading of RRM, in which whether a given encounter is hostile or benign is often accidental and the theoretical work is done by the way reciprocal tendencies and perceptions convert that encounter into a change in propensity; but it does mean that mRRM traces the dynamics of propensity rather than of offending, and that no crime rate is generated. Second, perceptions of the environment are held fixed, so mRRM cannot represent the environmental feedback on which RRM's population-level cycles depend;  it is, however, a reasonable simplification of propensity dynamics on sufficiently short time-scales (e.g., a day), over which perceptions are expected to remain roughly stable. Third, the parameters are not empirically calibrated: the results should be read as claims about the qualitative regimes the posited mechanisms can produce, not as quantitative predictions about any population. 

Taken together, the results stand in three different relations to RRM. Some recover propositions the theory already contains. That extremity of perception should beget stubbornness, and that those without a settled view of how they are generally treated should be the most exposed to the influence of whoever they encounter, are not artefacts of the formalism: both are properties of the account of environmental perception on which RRM was built \citep{svingen}. Their reappearance here is a check on the fidelity of the formalisation rather than a new finding. Others sharpen a proposition the theory states only loosely. RRM holds that reciprocal and retributive tendencies govern the responsiveness of propensity to social influence; mRRM specifies the form of that governance, showing that these parameters control the amplitude of fluctuation rather than its direction, and that they act upon the balance between displacement towards an extreme and return towards the perceived environment. The third class is the most consequential for the argument of this paper: results that RRM does not contain and could not have generated. That a population sharing a neutral perception of its environment divides while a population sharing a non-neutral one agrees, and that high cooperative tendencies -- high reciprocity and high retribution together -- should promote rather than restrain movement to the criminal extreme, are predictions of the formalised mRRM alone. They follow from mechanisms that RRM already posits but only once those mechanisms are stated precisely enough to be iterated, and they are available for empirical test in a way the verbal claims are not.

\section{Conclusion}\label{sec:conclusion}

We have developed an agent-based model of criminal propensity, extending the Deffuant framework of bounded confidence opinion dynamics to encapsulate the Retribution and Reciprocity Model of crime causation. Novelty in the mathematical model manifests as mechanisms for third-party witnesses of interactions to undergo propensity evolution, and for all agents to balance a pull to global extreme values with an attraction to individual perceptions of the environment. This balance results in oscillatory propensities being the dominant regime of model behaviour, demonstrating that persistent oscillations can arise in generalised bounded confidence models without adaptive bounds. The necessary condition for oscillation is satisfied across a broad region of parameter space, so that in a population with randomised parameter values only a minority of agents can be expected to be incapable of oscillating at all. Convergent behaviour also arises under certain sufficient conditions, which we have proven analytically, about the types and strengths of agent interactions. Polarisation, where a population splits into extreme factions, is found through computational simulations to arise when a population shares a neutral perception of the environment. In such systems, agents are amenable to change since it is likely for agents to find interactions sufficiently shocking that they move their propensities towards extreme values, and the influence of any agents holding outlying propensities in the population can exacerbate polarisation. Thus, within a narrow band around it, the very condition one might expect to produce consensus -- universal neutrality -- instead maximises division.

There are a number of avenues to extend the presented work. Oscillation currently arises when a propensity is alternately pushed towards an extreme and pulled back towards the individual perception of the environment, \(P_i\), which is held constant. Allowing the perception of the environment to evolve as \(P_i = P_i(t)\) would set in motion the fixed point identified by Theorem \ref{thm:converge2}. This conversion of a static attractor into a moving one may improve the mathematical model's correspondence with the free-rider dynamics theorised by \cite{svingen}. A simple starting point for modelling the dynamics of $P_i(t)$ may be to set it equal to the mean average of neighbouring nodes in a network structure. This formula follows the intuition that an agent's perception of the environment should depend on the criminal propensities within their social network. A more sophisticated approach would be to impose bounded-confidence dynamics on $P_i(t)$ in such a way that a timescale separation naturally emerges, where propensity dynamics occur much faster than perception dynamics. Another feature that can be incorporated is memory, allowing agents to account for the influence of past interactions when modulating their propensities, and it may be expected that agents with longer memories will be less likely to take on extreme values of propensity \citep{stokes}.

To our knowledge, this work is the first attempt to formalise a criminological theory of individual offending propensity within the opinion dynamics tradition, as existing agent-based work in criminology has modelled the spatial and situational distribution of crime events rather than the dynamics of propensity itself. The mathematical model is a novel extension to opinion dynamics, yielding a rich spectrum of nonlinear behaviour, particularly sustained oscillations through competition between social influence and fixed individual attractors. The modelling framework is adaptable and extensible,  allowing quantitative predictions to be generated from hypothesised causal mechanisms, and therefore providing an analytical foundation for testing criminological theory.

\end{document}